\documentclass{bioinfo}

\pdfoutput=1

\copyrightyear{2019}
\pubyear{2019}
\access{Advance Access Publication Date: Day Month Year}
\appnotes{Manuscript Category}

\usepackage{hyperref}
\usepackage{booktabs} % For formal tables
\usepackage[ruled]{algorithm2e} % For algorithms
\usepackage{amsmath}
\usepackage{amssymb}
\usepackage{ctable}
\usepackage{multirow}
\usepackage{graphicx}

\usepackage{enumitem}
\usepackage{setspace}
\usepackage{balance}
\theoremstyle{definition}
\newtheorem{definition}{Definition}[section]

\newcommand{\mytablespacing}{1}
\newcommand{\editage}{black}
\newcommand{\editageTwo}{black}

\newcommand{\revisionTwo}{black}
\newcommand{\revisionThree}{black}

\newcommand{\anna}{black}

\newtheorem{theorem}{Theorem}[section]

\begin{document}
\firstpage{1}

\subtitle{}

\title[AnomiGAN]{AnomiGAN: Generative adversarial networks for anonymizing private medical data}
\author[Bae \textit{et~al}.]{Ho Bae\,$^{\text{\sf 1}}$ Dahuin Jung\,$^{\text{\sf 2}}$ and Sungroh Yoon\,$^{\text{\sf1,\sf 2,3,}*}$}
\address{$^{\text{\sf 1}}$Interdisciplinary Program in Bioinformatics, Seoul National University, Seoul 08826, Korea and\\ $^{\text{\sf 2}}$Electrical and Computer Engineering, Seoul National University, Seoul 08826, Korea \\
$^{\text{\sf 3}}$Biological Sciences, Seoul National University, Seoul 08826, Korea
}
\corresp{$^\ast$To whom correspondence should be addressed.}
\history{Received on XXXXX; revised on XXXXX; accepted on XXXXX}
\editor{Associate Editor: XXXXXXX}

\abstract{
	Typical personal medical data contains sensitive information about individuals. Storing or sharing the personal medical data is thus often risky. For example, a short DNA sequence can provide information that can not only identify an individual, but also his or her relatives. Nonetheless, most countries and researchers agree on the necessity of collecting personal medical data. This stems from the fact that medical data, including genomic data, are an indispensable resource for further research and development regarding disease prevention and treatment. To prevent personal medical data from being misused,  techniques to reliably preserve sensitive information should be developed for real world application. In this paper, we propose a framework called anonymized generative adversarial networks (AnomiGAN), to improve the maintenance of privacy of personal medical data, while also maintaining high prediction performance. We compared our method to state-of-the-art techniques and observed that our method preserves the same level of privacy as differential privacy (DP), but had better prediction results. We also observed that there is a trade-off between privacy and performance results depending on the degree of preservation of the original data. Here, we provide a mathematical overview of our proposed model and demonstrate its validation using UCI machine learning repository datasets in order to highlight its utility in practice. Experimentally, our approach delivers a better performance compared to that of the DP approach.

	\textbf{Availability:} Implementation is available at \href{ http://data.snu.ac.kr/pub/XXX }{ https://github.com/hobae/AnomiGAN/ \\}
	\textbf{Contact:} \href{sryoon@snu.ac.kr}{sryoon@snu.ac.kr}
	%\textbf{Supplementary information:} Supplementary data are available at \textit{Bioinformatics}
	%online.
	}

\maketitle

\begin{methods}
    \section{Introduction}
    To restrain the use of personal data for illegal practices, the right to privacy has been introduced and is being adaptively amended. The right to privacy of medical data should be enforced because each individual's genetic information is static; therefore, a leak of such information would be very dangerous. Genetic markers, which are short DNA sequences, constitute a very sensitive piece of information. Using a genetic marker, it is feasible to uniquely identify individuals and their relatives. If proper security toward genetic information is not achieved, there is a risk of genetic discrimination such as denial of insurance (e.g., denial of insurance) or blackmail (e.g., planting fake evidence at crime scenes)~\citep{wagner2018technical}.

The advent of next-generation sequencing technology has progressed DNA sequencing at an unprecedented rate, thereby enabling impressive scientific achievements~\citep{schuster2007next}.
Using information gathered from the Human Genome Project, an international effort has been made to identify the hereditary component of the disease, which will allow for earlier detection and more effective treatment strategies~\citep{collins2001human}.

Data sharing between medical institutions is essential for the development of novel treatments for rare genetic diseases. Researchers are capable of identifying similar occurrence patterns of certain rare diseases on the basis of shared, more general data~\citep{weir2004stored}.
Therefore, seamless progress in genomic research largely depends on the ability to share data among different institutions~\citep{oprisanu2018anonimme}. 
\textcolor{\anna}{A public database GeneBank has been generated by international contributors under 
and is known as the Human Genome Project.}
\textcolor{\anna}{Various institutions are responsible for collecting health and genetic data in several countries.}
In the US, Research Program was launched to collect health and genetic data in 2015. In the UK, Genomics England sequenced the genomes of 100,000 patients with rare diseases and cancer, and the NIH's Genomic Data Commons (GDC) serves as a unified data repository from the cancer research community (https://gdc.cancer.gov/).

Patients portals and telehealth programs have gained popularity to patients allowing them to interact with their healthcare system by online health services~\citep{crotty2016designing}. Although these online health services provide convenience since patients can order prescriptions at home or remotely, patients are required to transmit their private data over the Internet.
\textcolor{\anna}{Most health services follow the guidelines of accountability act of 1996 (HIPPA\footnote{*The HIPAA statute that Genetic Data, by definition linked to an identifiable person, should not be disclosed or made accessible to third parties, in particular, employers, insurance companies, educational institutions, or government agencies, except as required by law or with the separate express consent of the person concerned.}\label{fn:HIPPA}) to protect patient records, but these guidelines may not be upheld when data is shared to a third party.}

Online health services are extremely useful tools, but have introduced vulnerabilities related to privacy issues. There are two significant privacy threats involved in online health services: a) hijacking during transmission and b) privacy leak on storage. To encourage the use of online health services and the provision of sensitive genetic information, a strong privacy shield should be guaranteed for all users and donors, both during transmission and storage.

There are two approaches to improve protection of personal medical data: 1) statistical-based anonymization and 2) encryption. The statistical-based approach relies on strong assumptions of the background population~\citep{sean2019}. Differential privacy (DP)~\citep{dwork2011differential} is a state-of-the-art method used to provide strong privacy guarantees. \textcolor{\anna}{In addition to DP,} a number of methods have been suggested to enable the sharing of aggregate personal medical data while preserving participants privacy~\citep{erlich2014routes,homer2008resolving,zhou2011release,sankararaman2009genomic,simmons2015one}.
However, the first approach is limited to a few genomic loci, which can lead to inaccuracies when the size of the genomic loci is scaled up~\citep{simmons2016realizing,simmons2016enabling}.

Second, cryptography-based methods (homomorphic encryption) enables the computation of encrypted data via simple operations, such as summation and multiplication. The property of homomorphic encryption allows for sharing of \textcolor{\anna}{personal medical} data with accurate results. 
However, the latency of the computation is in the order of hundreds of seconds~\citep{gilad2016cryptonets}. In addition, a simple operation of homomorphic properties limits adaptation to a complex model such as neural networks~\citep{bae2018security}.

As a number of deep learning-based disease prediction tools have been developed and demonstrate outstanding prediction results~\citep{min2017deep}, the deep learning techniques are extended to design privacy-preserving deep neural networks~\citep{gilad2016cryptonets, hesamifard2017cryptodl, sanyal2018tapas,kim2016collaborative}.
For example, cryptoNets~\citep{gilad2016cryptonets} apply neural networks to infer encrypted data such as DNA sequences. However, the accuracy and efficiency of this model are very low because the activation functions are replaced by non-polynomial activation functions and the converted precision of the weights~\citep{bae2018security}. Recently, Sanyal et al. modified CryptoNets in a way that allows for parallelization~\citep{sanyal2018tapas}.

The first approach to privacy preserved sharing does not provide a computational overhead, but this method has a significant trade-off between privacy and accuracy. The second approach guarantees prediction accuracy while preserving privacy, but contains a bottleneck for the server-side computation-complexity. To address the aforementioned issues, we propose a method based on generative adversarial networks (GANs) that preserves similar prediction performance.

\subsection{Problem Statement}
We focus on privacy issues related two types of adversaries that may be encountered
through the use of online health services: active and passive adversaries. The active adversary will target an individuals private data while users are interacting on the fly with the service, and the passive adversary will target any information that is stored via online or offline services. 
For the passive adversary, two major privacy issues may be involved: a) a private leak 
occurring because of the level of security of the service system, and b) upon patient's consent for the use of their medical information for the research purposes. The patients record may be propagated to a third party with minimum anonymization following the deidentification guideline~\citep{berhane2012concepts}.

 \begin{figure*}
\centering
\includegraphics[width=0.795\textwidth]{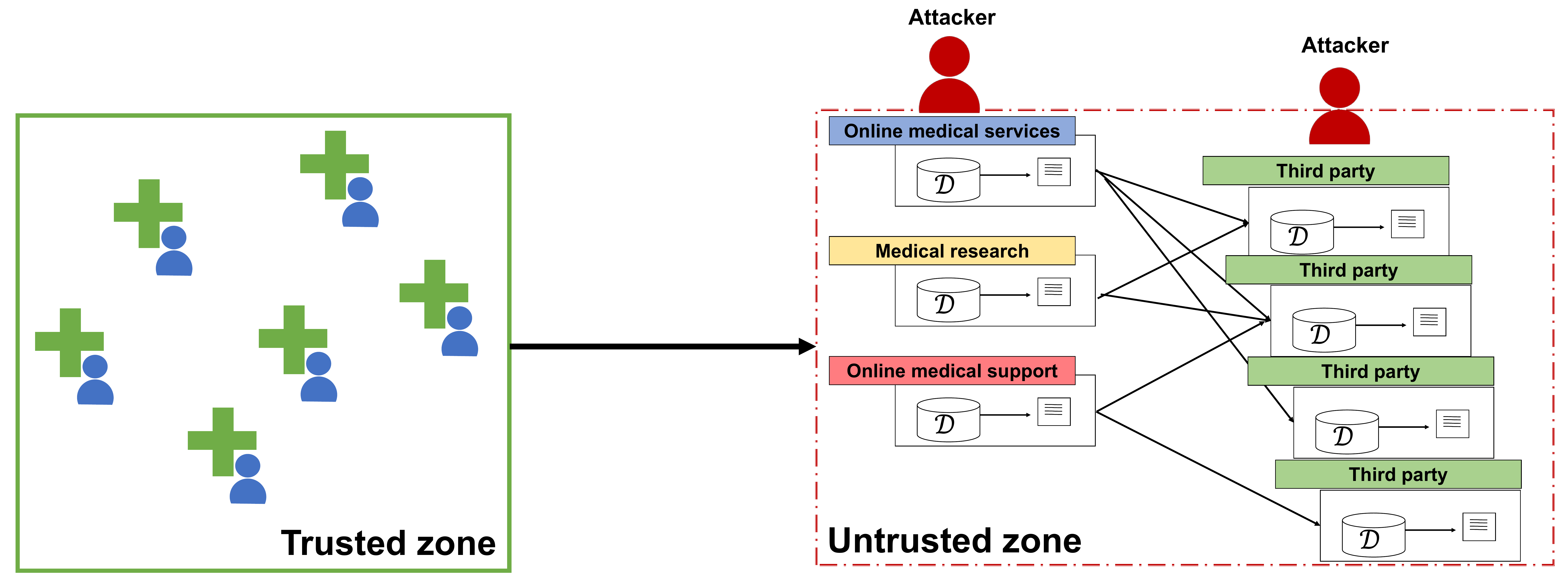}
\caption{The scenario of this study. The scenario consists of a trusted zone and an untrusted zone; the untrusted zone can be further divided into two groups. The first group is comprised of online medical services, and the second group includes third parties (Google, Dropbox, and Amazon). User's medical data is transferred to the online medical service and the services provide diagnosis results to the user. Upon user consent for data sharing, the user's data may be propagated to the third parties.}
\label{fig:overview}
\end{figure*} 

\subsection{Solution Intuition}
Our proposed framework, AnomiGAN, allows a user to control the anonymization level. For strong anonymization, one parameter that controls the privacy level can be minimized to zero.
Then, the model will anonymize the data purely depending on a target classifier that guarantees the same prediction results as those for original data while preserving participant's privacy.

AnomiGAN also provides functionality that mitigates the level of confidentiality. This functionality can be used to ensure minimal anonymity when sharing data between institutions that follow the same privacy guidelines. The confidence level will depend on the $\delta$ acting like the $\delta$-differential privacy method~\citep{dwork2013toward}. 

Unlike other deep neural networks, AnomiGAN's design is a non-deterministic model developed by adding variance which is estimated by the parameters that are stored during the training process to the random layer of the trained model.

\subsection{Contributions}
In this paper, we propose the anonymous generative adversarial networks to protect an individual's disease information from the institutions with the ability to aggregate data among different institutions. Our framework is a generic method that exploits a target classifier simulating prediction efforts to preserve the original prediction result. 
We explore whether a generative model can be constructed to produce meaningful synthetic data, which simultaneously preserves original information while protecting private disease information.
We evaluated the proposed method on prediction classifiers for two diseases (breast cancer, chronic kidney disease), and found that prediction performance degradation is minimal compared to the original prediction result.
Finally, we compared our proposed method to the state-of-the-art privacy preserving technique, and provide analysis regarding privacy parameter and accuracy.
Our analysis of the anonymous generative adversarial networks reveals that the model becomes unstable if the model employs a non-robust, target classifier.

The remainder of this paper is organized as follows.  In the following section, we define entities, operations and adversary's objectives. In Section 3, we present a mathematical overview of our proposed model with architecture.
In Section 4, we present our experimental results and compare with DP. Finally, we conclude our paper in Section 5.
\end{methods}

\begin{methods}
    \section{Background}

 We assume that the service providers use supervised machine learning classifiers to make predictions on personal medical data. A machine learning-based classifier attempts to find a function $f$ that maps a medical data points into either benign or malicious. Details of the target classifiers are described in Section~\ref{target_classifer}. In the following section, we define the goal and capabilities of adversaries, and their associations to security bound. We provide preliminary security definition in  Section~\ref{threat_model} for the proof of our model described in Section~\ref{model_proof}.

\subsection{Scenario of Our Study}
\figurename~\ref{fig:overview} illustrates the overall workflow of our study. The scenario consists of a trusted zone and an untrusted zone. In the trusted zone, patients have not involved any threats from an adversary.
Within the untrusted zone, the scenario assumes that two groups exist: a) one group that follows the health insurance portability and accountability act of 1996 (HIPPA) guidelines to protect patient's records, and b) the group that follows guidelines applicable to Cloud services, Google, Dropbox, Amazon, etc. Upon the user's consent for data sharing, the user's data may be propagated to third parties. Two possible types of adversaries exist when private information is sent from the trusted zone to an untrusted zone for online health services. 
Private information is a candidate target for an active adversary during transmission, and any portion of an individual's private information that is stored is also a suitable target for a passive adversary.
To prevent any compromise of sensitive information, AnomiGAN can be deployed in the trusted zone to anonymize personal medical data.

\subsection{Adversarial Goal and Capabilities}
It is important to define the capability of a particular adversary to measure a relevant privacy aspect.
The adversarial' goal in our case is to compromise an individual's private medical data including any sensitive information. 
An adversary can be anywhere in both online/offline health services, and any third parties that closely work with medical institutions. 

An adversary often makes an effort to estimate a posterior probability distribution with the resources available for breaching privacy. The available resources can be multiple combinations of computation power, time, bandwidth, or physical nodes~\citep{wagner2018technical}. The adversary's success probability can be quantified with the many trials of adversary's choice of input. The probability is then used to quantify privacy, with low probability correlating to high privacy.

\subsection{Security Bound}
Modern cryptography introduces two relaxations to the notion of perfect security~\citep{lindell2014introduction}. The first notion limits the adversary to the polynomial time adversaries indicating that security is only guaranteed against polynomial-time adversaries. The second notion relaxes the adversary's success probability to allow for a small (negligible) advantage. As such, we have designed our scheme as a probabilistic polynomial-time algorithm, whereby the output of the model must have randomness to prevent an adversary who repeated procedure with same input to observe an information leak. 
We assume that adversary run in polynomial time with additional negligible winning probability. For example, let $\mathbf{A}$ be an algorithm that runs in polynomial time $p(\cdot)$  for every input $x \in \{0,1\}^{*}$. The computation of $\mathbf{A}(x)$ is then within at most $p(|x|)$ steps. Note the probabilistic algorithm is in the cryptographic system can be viewed as having a capability of tossing coins. Tossing coins means that the algorithm $\mathbf{A}$ has access to the random oracle such that each time of tossing coins will independently equal to 1 with probability $1/2$ and to 0 with probability $1/2$.

\subsection{Generative Adversarial Networks}
Generative Adversarial Networks (GANs)~\citep{goodfellow2014generative} are designed to solve other generative models by introducing a new concept of adversarial learning between a generator and a discriminator instead of maximizing a likelihood. The generator produces real-like samples by transformation function mapping a prior distribution from latent space into the data space. The discriminator acts as an adversary to distinguish whether samples generates from the generator derive from the real data distribution.
Notably, the application of GANs has extended to various fields of studies. 
For example, GANs have recently been employed in cryptography and steganography~\citep{baluja2017hiding}.

\subsection{Differential Privacy}
DP is a privacy-preserving model~\citep{dp_survey} that guarantees to protect individuals from privacy loss from an adversary. Intuitively, DP promises that the probability of harm can be minimized by adding noise to the output as follows:
\begin{align}
M(\mathcal{D}) = f(\mathcal{D}) + \tau
\end{align}
where $M: \mathcal{D} \rightarrow \mathbb{R}$ is a random function that takes a noise $\tau$ to the output, $\mathcal{D}$ is the target database, and $f$ is the deterministic original-query response.

\begin{definition}
(Differential Privacy). A random algorithm $M$ with domain $\mathbb{N}^{|D|}$ is $\delta$-differentially private if for all $S \in$ Range$(M)$ and for all $D,\hat{D} \in \mathbb{N}^{|D|}$ such that $D - \hat{D} \leq 1$:
\begin{align}
\textrm{Pr}[M(\mathcal{D}) \in S] \leq \exp ({\delta}) \textrm{Pr}[M(\mathcal{D}')\in S]
\label{eq:dp_def}
\end{align}
where $\mathcal{D}$ and $\mathcal{D}'$ are the absolute value of the privacy loss that is bounded by $\delta$ with probability at least $1-\delta$~\citep{dwork2014algorithmic}.
\end{definition}

\iffalse
DP is a privacy-preserving model~\citep{dp_survey}; this model guarantees with high confidence that an adversary cannot infer any private information from databases or released models.
DP algorithms counter privacy threats by adding noise to the response as follows:
\begin{align}
M(\mathcal{D}) = f(\mathcal{D}) + \tau
\end{align}
where $M: \mathcal{D} \rightarrow \mathbb{R}$ is a randomized function that combines the noise $\tau$ to
the query response, $\mathcal{D}$ is the target database, and $f$ is the original query response, which is deterministic. 
$M$ gives $\delta$-DP if all adjacent $\mathcal{D}$ and $\mathcal{D}'$ satisfy the following:
\begin{align}
\textrm{Pr}[M(\mathcal{D}) \in S] \leq \exp ({\delta}) \textrm{Pr}[M(\mathcal{D}')\in S]
\label{eq:dp_def}
\end{align}
where $\mathcal{D}$ and $\mathcal{D}'$ are two adjacent databases and $S \subseteq \mathrm{Range}(M)$ is a subset of $\mathbb{R}$. $\delta$ is the privacy budget parameter that decides the privacy level.
\fi

\subsection{Notations}

The notations used in this paper are as follows:

\begin{itemize}[leftmargin=5.5mm, labelindent=7.5mm,labelsep=3.3mm]

\item $x$ is the input.

\item $\hat{x}$ is the anonymized output given input $x$.

\item $E$ is an encryption function. $E$ takes input $x$ and returns the encrypted output $E(x) \rightarrow \hat{x}$. 

\item $F$ is a random generator. $F$ takes a seed $k$ and returns a random output string $F(k) \rightarrow r$

\item $\mathbf{M}$ is a trained model.

\item $y$ is an output \textcolor{\revisionTwo}{score} given by the trained model $\mathbf{M}(x) \rightarrow \textcolor{\revisionTwo}{y}$ given input $x$.

\item $\hat{y}$ is an output score given by the trained model $\mathbf{M}(\hat{x}) \rightarrow \textcolor{\revisionTwo}{\hat{y}}$ given input $\hat{x}$.

\item $\mathcal{A}$ is a probabilistic polynomial-time adversary. The adversary is an attacker that queries input to the oracle model.

\item $\epsilon$ is the standard deviation value of \textcolor{\revisionTwo}{score} $y$.

\item $\delta$ is a privacy parameter that controls confidence levels.

\end{itemize}

% anna check-up part
\subsection{Threat Model}\label{threat_model}
A random oracle model~\citep{canetti2004random} posits the randomly chosen function $H$, which can be evaluated only by \textcolor{\editageTwo}{querying} an oracle that returns $H(x)$ given input $x$.
The security of the random oracle is based on an \textcolor{\revisionTwo}{\emph{experiment}} involving an adversary $\mathcal{A}$, \textcolor{\editageTwo}{as well as} $\mathcal{A}$'s indistinguishability of the encryption.
Assume that we have the random oracle that acts like a current anonymous scheme $E$ with only a negligible success probability.

The \emph{experiment} can be defined for any encryption scheme $E$ over input space $\textbf{X}$ and for adversary $\mathcal{A}$. The \emph{experiment} is defined as follows:

\begin{enumerate}  [leftmargin=*,label=(\roman*),labelindent=1.5mm,labelsep=1.3mm]
    \item The random oracle chooses a random anonymous scheme $E$.
    Scheme $E$ modifies or extends \textcolor{\editage}{the process of} mapping a medical data $x$ \textcolor{\editageTwo}{with} length $n$ to transformed medical medical data $\hat{x}$ as the output. 
    \textcolor{\editageTwo}{The} process of mapping sequences can be \textcolor{\editageTwo}{considered} as a table \textcolor{\editageTwo}{that} indicates for each possible input $x$ the corresponding output value $\hat{x}$.

    \item Adversary $\mathcal{A}$ then chooses a pair of medical data \textcolor{\revisionThree}{ $x_0,x_1$ }.

    \item The random oracle \textcolor{\editageTwo}{selects} a bit \textcolor{\revisionThree}{ $b \in \{0,1\}$} and sends encrypted medical data $ E(x_b) \textcolor{\revisionTwo}{\rightarrow} \hat{x}$ to the adversary.

    \item The adversary outputs a bit $b'$.

    \item The output of the \emph{experiment} is defined as 1 if $b' = b$, and 0 otherwise. $\mathcal{A}$ succeeds in the \emph{experiment} in the case of distinguishing $x_b$.

\end{enumerate}

\begin{figure*}
\centering
\includegraphics[width=0.875\textwidth]{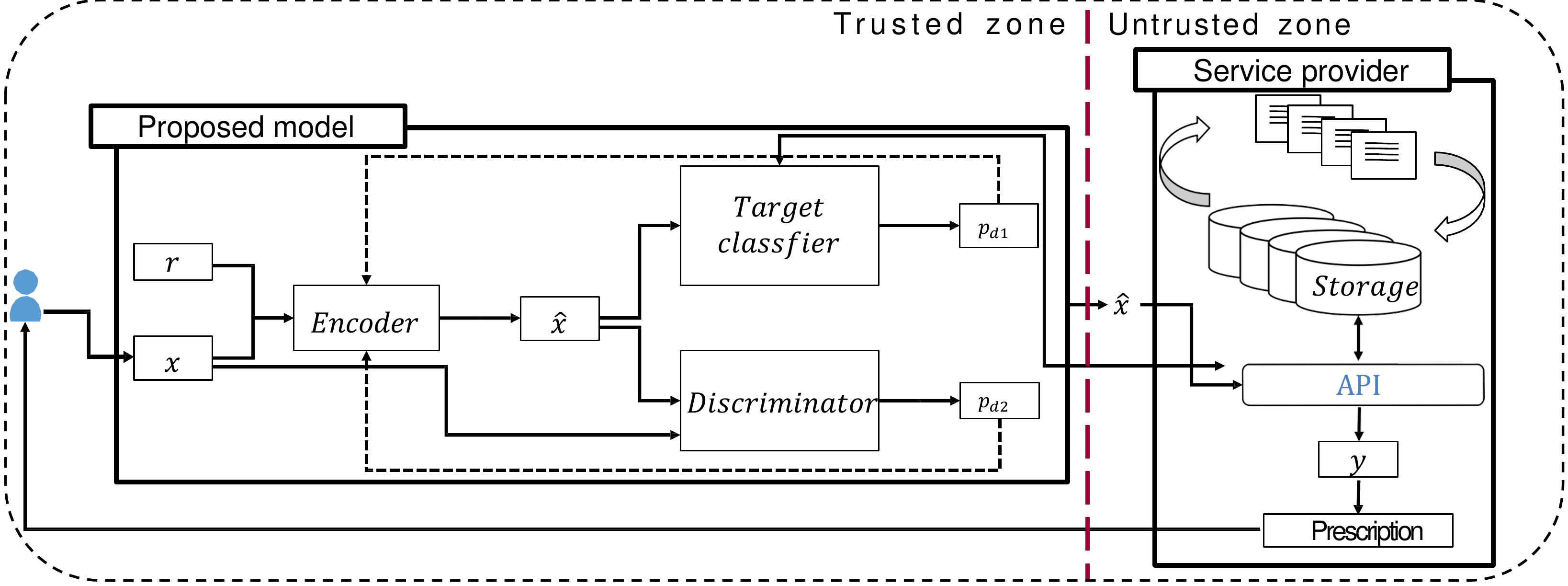}
\caption{Architecture of the model presented in this study.}
\label{fig:architecture}
\end{figure*}

With the \emph{experiment}, the definition of perfect security for $E$ takes the following general form\textcolor{\revisionTwo}{~\citep{canetti2004random}}:
\begin{definition}
The scheme $E$ is perfectly secure over input space $\textbf{X}$ if for every adversary $\mathcal{A}$ it satisfies \begin{equation}
    \text{Pr}[\text{experiment \ of \ success}] = {1 \over 2}.
\end{equation}
\label{perfectDef}
\end{definition}
In the encryption scheme $E$, $\mathcal{A}$ cannot distinguish between \textcolor{\revisionThree}{$x_0$ and $x_1$}. \textcolor{\editageTwo}{Furthermore,} $\mathcal{A}$ obtains no information about the presence of a hidden message.
In the real world, most systems do not have access to a random oracle. Thus, pseudorandom function\textcolor{\editageTwo}{s are} typically applied by replacing the random function\textcolor{\revisionTwo}{~\citep{canetti2004random}}.
With \textcolor{\editageTwo}{this} assumption, the oracle is replaced by a fixed encryption scheme $E$, which corresponds to the transformation of a real system (implementation of the encryption scheme).
The implementation of a random oracle is deemed secure if the \textcolor{\editageTwo}{probability of} the success of a random oracle \textcolor{\editage}{attack} is negligible. Moreover, the encryption scheme $E$ is soundness secure if adversary $\mathcal{A}$ has a $success$ probability such that

\begin{equation}
\text{Pr}[\text{success}] \leq {1 \over 2} + \epsilon.
\label{win_equation}
\end{equation}

\end{methods}

\begin{methods}
    \section{Methods}\label{method}
    \begin{figure*}
\centering
\includegraphics[width=0.95\textwidth]{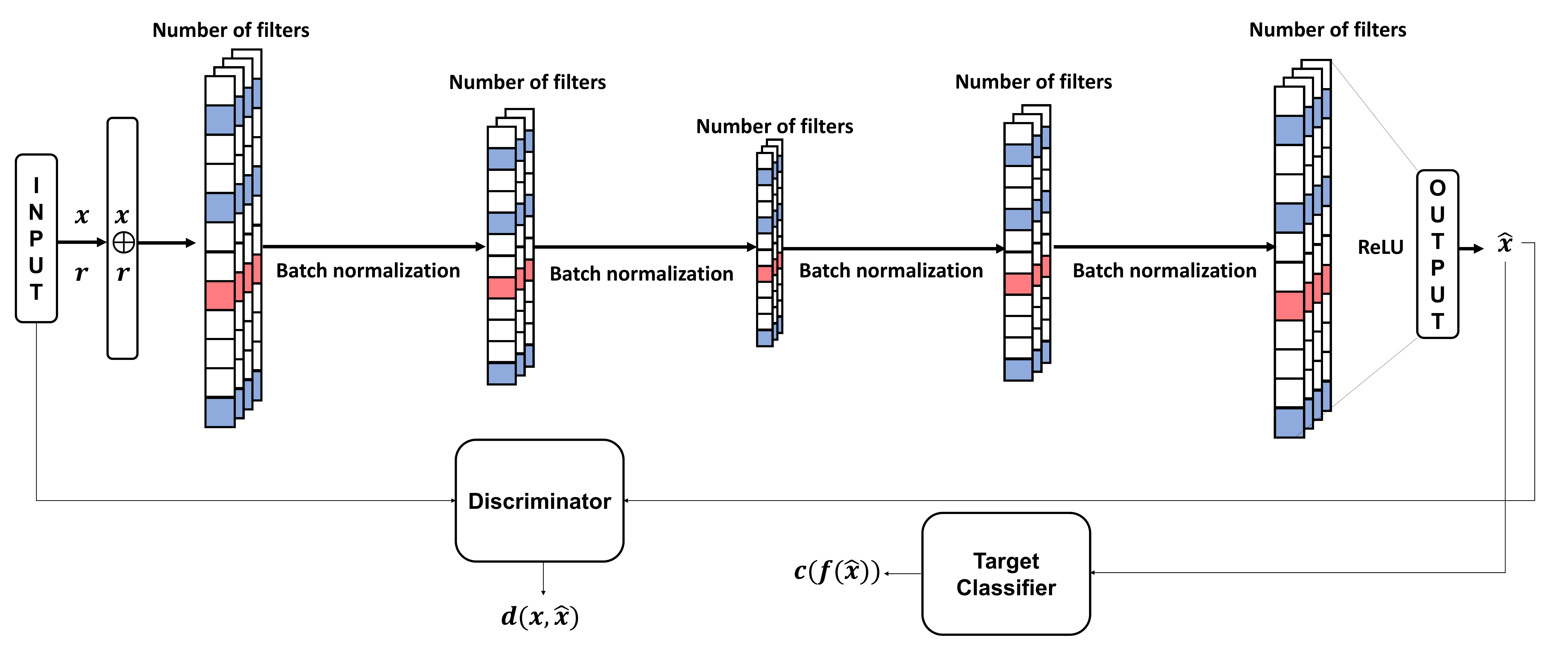}
\caption{Model training. The encoder accepts $x$ and $r$ as input and that are fed into the neural network. The discriminator takes an original input and output of the encoder to output probabilities from the logits (last fully connected layer). The target classifier takes an input $\hat{x}$ and outputs the prediction score. }
\label{fig:encoder}
\end{figure*}

% anna check-up
Our training involves three parties: an encoder, a discriminator, and a target classifier (pre-trained discriminator). The encoder generates synthetic data close to the form of input data, and the target classifier places a prediction score of a synthetic data generated by the encoder. The discriminator then generates a confidence score of whether a piece of data is a synthetic or original data. The encoder is trained with random noise to learn to generate a synthetic data such that the prediction result of a synthetic data from a target model is same as the original data.

\subsection{Anonymization using GANs}
The architecture of this model is illustrated in \figurename~\ref{fig:architecture}. 
The encoder takes an input $x$ and gives the output $\hat{x}$, which are given to both the discriminator and the target classifier. The discriminator outputs the probability that $x = \hat{x}$, given input $\hat{x}$. The target classifier outputs a prediction score given input $\hat{x}$. The learning objective of the encoder is to minimize discriminator's probability to $1 / 2$ and maximize the prediction score of the target classifier.

The encoder accepts the $n$ length of messages as input and generates the $n$ length of random strings, that are fed into the neural network. The first layer consists of 64 filters with a length of 4; The second layer consists of 32 filters with a length of 2; The third layer consists of 16 filters with a length of 2. The fourth layer consists of 8 filters with a length of 2. Additional layers are then additional layers are added in the reverse order to the number of input length $n$. Batch normalization~\citep{ioffe2015batch} is used at each layer and tanh~\citep{lecun2012efficient} is used as the activation at each layer, except for the final layer where ReLU~\citep{nair2010rectified} is used for the activation function. The discriminator takes the output of the encoder as an input to determine whether the output is real or generated. A sigmoid activation function is used to output probabilities from the logits.

To define the learning objective, let $\theta_{E}$, $\theta_{D1}$, and $\theta_{D2}$ denote parameters of the encoder, discriminator and target classifier. Let $E(\theta_{E}, x, r, \delta)$ be the output on $x$, 
$D2(\theta_{D2}, \hat{x})$ be the output on $\hat{x}$, and $D1(\theta_{D1}, x, E(\theta_{E}, x, r, \delta), \theta_{D2})$ be the output on $x$ and $\hat{x}$.
Let $L_{E}, L_{D1}, L_{D2}$ denote the loss of encoder, discriminator, and target classifier.
Let $\lambda_{e}$ and $\lambda_{d}$ denote the weight parameters for the encoder, the discriminators. The parameter $\lambda_{e}$ can be used to control anonymization level. For strong anonymization level, the parameter $\lambda_{e}$ can be the minimum value.
The encoder then has following objective functions:
\begin{equation}
\begin{aligned}
         \mathrm{L}_{E}(\theta_{E}, x, r, \delta) & = \lambda_{e} \cdot d(x, E(\theta_{E}, x, r, \delta)) + \lambda_{d} \cdot ( \mathrm{L}_{D1} + \mathrm{L}_{D2}) \\
    & = \lambda_{e} \cdot (d(x, \hat{x}) + \delta) +  \lambda_{d} \cdot (\mathrm{L}_{D1} + \mathrm{L}_{D2})
\end{aligned}
\label{loss1_eq}
\end{equation}
where $d(x, \hat{x})$ is the Euclidean distance between synthetic and original data and $\delta$ controls confidence level. 
A discriminator has the sigmoid cross entropy loss of:
\begin{equation}
\begin{aligned}
    \mathrm{L}_{D1}(\theta_{E}, \theta_{D1}, \theta_{D2}, x, \hat{x}) = 
      &  -y \cdot \mathrm{log}(E(\theta_{E}, p)) \\
      &  - (1-y) \cdot \mathrm{log}(1 - E(\theta_{E}, p)),
\end{aligned}
\end{equation}
where $y=0$ if $p=\hat{x}$ and $y=1$ if $p=x$, where $p$ is the score of given input $x$ and $\hat{x}$.
A target classifier has a classification score of:
\begin{equation}
      \mathrm{L}_{D2}(\theta_{D2}, \hat{x}) = C(f(\hat{x})),
\end{equation}
where $C(f)$ is a cost function of pre-defined classifier.

\subsection{Security Principle of Anonymized GANs}\label{model_proof}
In this section, we show that the AnomiGAN has a scheme that is indistinguishable for an adversary. The AnomiGAN has a similar structure to the one-time pad encryption scheme, except that a probabilistic model is used to generate the output of the one-time pad. 
Since a probabilistic model generates a pseudorandom output that appears to be random any polynomial-time adversary, the AnomiGAN can be demonstrated as computationally-secure.
A simple intuition of the indistinguishable scheme is that the adversary is allowed to choose from multiple data from the synthesized data. An adversary can freely interact with an encryption oracle, which is regarded as a black-box that anonymize data chosen by the adversary. For the AnomiGAN, synthesized data can be viewed as encrypted data. 
Formally, we have constructed as input a seed $k \in \{0,1\}^n$ and a medical record $x \in \{0,1\}^n$, where $r \leftarrow \{0,1\}^n$ is chosen uniformly at random and outputs the synthesized data.
We define our anonymization scheme as follows for medical record length $n$:

\begin{itemize}[leftmargin=5.5mm, labelindent=7.5mm,labelsep=3.3mm]

\item Select $k \leftarrow \{0,1\}^n$ and outputs the seed.

\item On input a seed $k$, pseudorandom generator $F$ outputs a random string $F(k) \rightarrow r$.

\item AnomiGAN: on the input a random string $r$, and a medical record $x$, the model $\mathbf{M}$ output the synthesized medical record
\begin{equation}
    \hat{x} = \mathbf{M}(r \oplus x).
\end{equation}
\end{itemize}
A random $r$ is replaced for each learning steps, and the variances of each layer are stored during the learning process. The variance of each layer is added to randomly selected layers in inference time to ensure that the generator does not produce the same output from the same input.
Intuitively, a generative model is a probabilistic model; thus $\mathbf{M}(r \oplus x)$ appears completely random to an adversary who observes a medical record $\hat{x}$ given $x$ operating similarly to the one-time pad. Note that similar operations can be expected upon the replacement of $ \oplus $ to $(r \times x)$ mod $n$.
\begin{theorem}
If $F(k)$ is a pseudorandom generator and $\mathbf{M}$ is a probabilistic model, then $\mathbf{M}$ has a scheme that is indistinguishable to an adversary.
\end{theorem}

\begin{proof}
The rationale for the proof is that if $\mathbf{M}$ is a probabilistic and the $F(k)$ is the pseudorandom generator; then the resulting scheme is identical to the one-time pad encryption scheme hold a Definition~\ref{perfectDef}.
We know that $r$ has the same length of $x$ and the output length of $\hat{x}$ is also equal to both $r$ and $x$. 
The same length of $r$, $x$, and $\hat{x}$ with the operation of $r \oplus x$ is identical to the one-time pad encryption scheme, which has a formal Definition~\ref{perfectDef}. Let polynomial-time adversary $\mathcal{A}$ constructs a distinguisher for $F(k)$ such that $\mathcal{A}$ has the success probability as defined in Equation~\ref{win_equation}. The distinguisher is given an input $x$, and the goal is to determine whether $x$ is a truly random or $x$ is generated by $\mathbf{M}$. The distinguisher emulates the $experiment$ described in Section~\ref{threat_model}, and the distinguisher has two observations.
If input $x$ is truly random, then the distinguisher has a success probability of:
\begin{equation}
   \text{Pr}[{F(k)} \text{\ of \ success}] = {1 \over 2},
\end{equation}
which follows Equation~\ref{perfectDef}.
If input $x$ is equal to $ c = \mathbf{M}(F(k) \oplus x) $, where $k$ is chosen at uniformly random, then the distinguisher has a success probability of:
\begin{equation}
\text{Pr}[\mathbf{M}(F(k) \oplus x) \text{\ of \ success}] \leq {1 \over 2} + \epsilon.
\label{win_equation}
\end{equation}
By the assumption that $F(k)$ is a pseudorandom generator and $\mathbf{M}$ is a probabilistic, then $\epsilon$ must be negligible. \qed
%Note that this assumption holds only if $k$ and pseudorandom function $\mathrm{F}$ is applied to a different input each time.
\end{proof}

\end{methods}

\begin{methods}
    \section{Results}\label{results}
    
\iffalse
\begin{figure*}
\centering
\includegraphics[width=1\textwidth]{figures/heatmap_both.pdf}
\caption{Feature correlation (best viewed in color).
(a) Correlations between features of breast cancer dataset.
(b) Correlation between feature of chronic kidney disease.
}
\label{fig:heatmap_both}
\end{figure*}
\fi

\begin{figure*}
\centering
\includegraphics[width=0.775\textwidth]{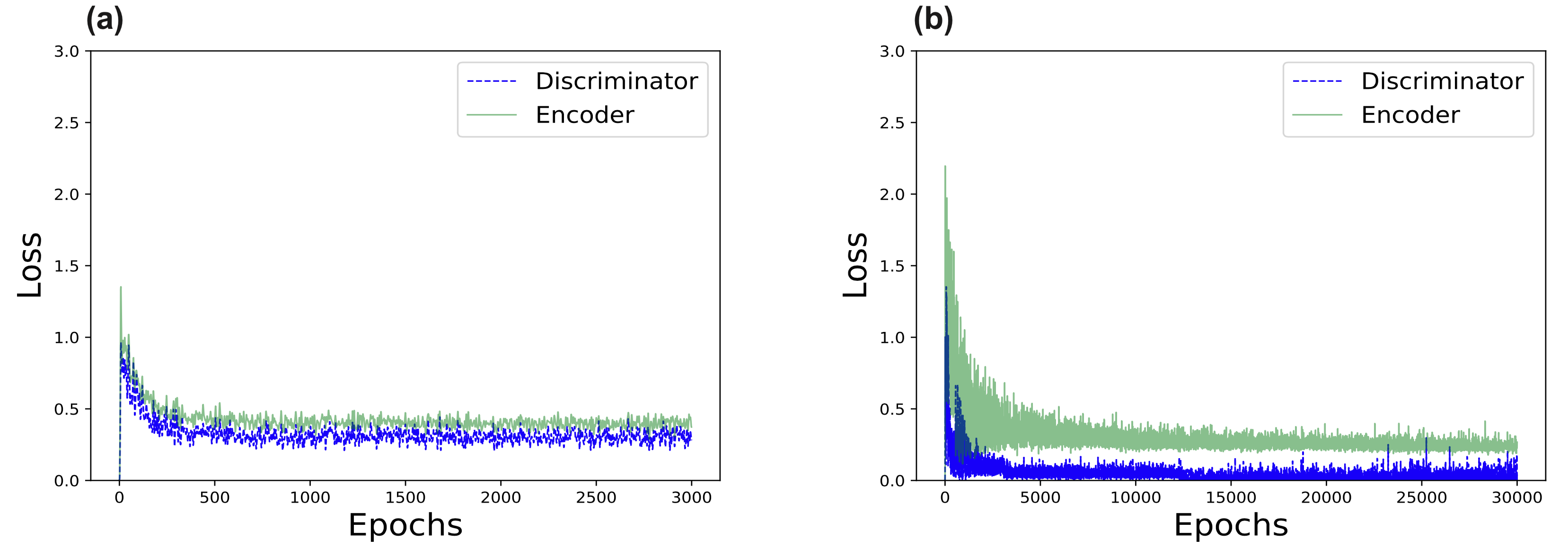}
\caption{Learning performance of AnomiGAN.
(a) Loss of encoder and discriminator for breast cancer dataset.
(b) Loss of encoder and discriminator for chronic kidney disease dataset.
}
\label{fig:boss_loss}
\end{figure*}

\subsection{Experimental Environment}
We performed experiments using Ubuntu 14.04 (3.5GHz Intel i7-5930K and GTX Titan X Maxwell(12GB)). For the implementation, we exploited the Scikit-learn library package (version 0.18) for convolutional neural networks, the Keras library package (version 2.0.6) for neural networks, tensorflow (1.11.0) for generative adversarial networks, and bio-python (1.72) used for input representation.

\subsection{Datasets}

We simulated our approach using the Wisconsin breast cancer dataset from the UCI machine learning repository~\citep{blake1998uci}, and the chronic kidney disease dataset from the UCI machine learning repository~\citep{rubini2015generating}. The Wisconsin breast cancer and chronic kidney disease datasets consist of 30 and 24 features, respectively. The datasets are randomly partitioned to training and test sets of 90\% and 10\%, respectively.

\subsection{Target Classifiers}\label{target_classifer}
Many services are incorporate disease classifiers using machine learning techniques. For our experiment, we selected breast cancer and chronic kidney disease model from the kaggle competitions as the target classifiers. The classifiers will be used as black-box access to our target classifier in our method. We selected these classifiers for two reasons: a) both classifiers achieve high accuracy in disease detection to their testing datasets, and b) these classifiers are open source implementations, which allows easily accessible as our target classifier.
%The oracle is used to determine if a generated sample is distinguishable to the original sample. 

\subsection{Model Training}
For the training model, we used optimizer of multi-class logarithmic loss function Adam \citep{kingma2014adam} with a learning rate of 0.001, a beta rate of 0.5, the epoch of 50000, and mini-batch size of 10.
The objective function $\mathcal{L}_E$ that must to be minimized as described in \textcolor{\revisionThree}{Eq~(\ref{loss1_eq})}. 
Most of these parameters and networks structure were experimentally determined to achieve optimal performance. \figurename~\ref{fig:boss_loss} show the training loss of each model. 
For the breast cancer dataset, the discriminator achieves the optimal loss after 3000 steps, while the encoder requires more steps to generate original data like the sample. For the chronic kidney disease dataset, the optimal loss achieved after 10,000 steps.

\subsection{Evaluation Process}
We exploited DP, in particular, the Laplacian mechanism~\citep{dwork2013toward}, to compare the anonymization performance against corresponding accuracy and AUC. For the evaluation metric, the accuracy\footnote{$\textrm{Accuracy} = (TP+TN) / (TP+TN+FP+FN)$, where $TP$, $FP$, $FN$, and $TN$ represent the numbers of true positives, false positives, false negatives, and true negatives, respectively.} and the AUC are used to measure performance between original samples and anonymized samples according to the model's parameter changes. The correlation coefficient is used to measure the linear relationship between the original samples and anonymized samples by changing the privacy parameters. We generated the anonymized data according to privacy parameter $\delta$ and $\lambda_e$ by randomly selecting 1000 cases, and obtained the average prediction of accuracy, AUC, and correlation coefficient against corresponding original data. 
In the next step, we fixed test data and generated anonymized data to validate the probabilistic behavior of our model. As shown in Table~\ref{table:cc_layers}, a variance of each encoder layers is added to the corresponding encoder layers in the inference time. The process was repeated 1000 times with the fixed test data.

\begin{figure*}
\centering
\includegraphics[width=1\textwidth]{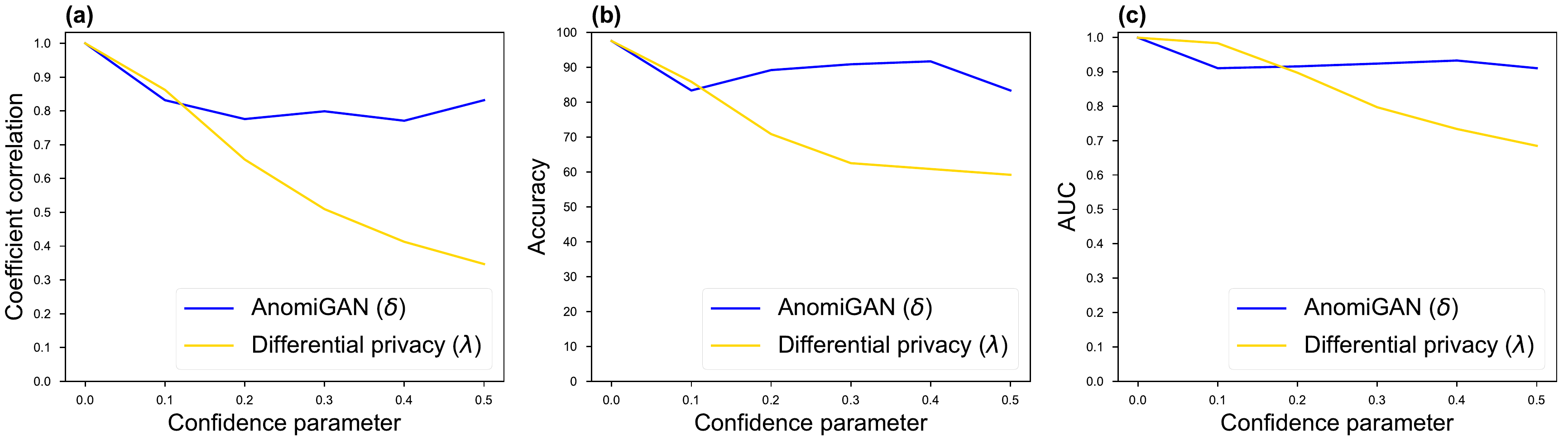}
\caption{Anonymization performance using a breast cancer dataset: a fixed test dataset was selected from the UCI machine learning repository. The correlation coefficient, accuracy, and AUC were measured by changing 0.1 of the privacy parameter for fixed test data,  $\delta$.}
\label{fig:results_delta_breast}
\end{figure*}

\begin{figure*}
\centering
\includegraphics[width=1\textwidth]{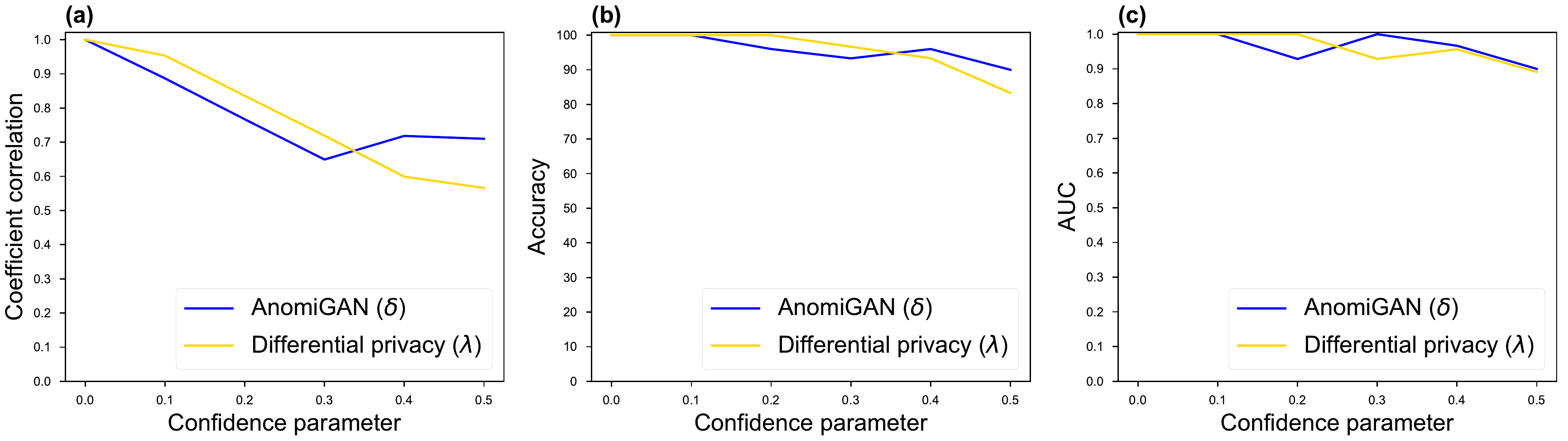}
\caption{Anonymization performance using a chronic kidney dataset: a fixed test dataset was selected from the UCI machine learning repository. The correlation coefficient, accuracy, and AUC were measured by changing 0.1 of the privacy parameter for fixed test data,  $\delta$.}
\label{fig:results_delta_kidney}
\end{figure*}

\subsection{Comparison to Differential Privacy}
DP achieves plausible privacy by adding Laplacian noise $\mathbf{Lap}(\lambda) = \mathbf{Lap}(S / \delta)$ to a statistics~\citep{dwork2013toward}.
The parameter $\lambda =.1$ has a minimal effect on privacy and the risk of privacy increases as the parameter $\lambda$ increases. The amount of noise presents a trade-off between accuracy and privacy. Note that standard DP of unbounded noise version of Laplacian was applied for the experiments. 

\figurename~\ref{fig:results_delta_breast} shows an experiment for our proposed algorithm and DP algorithm using a fixed breast cancer tests. The experiments were conducted by increasing 0.1 of the parameter $\lambda$. DP showed good performance in the coefficient correlation but showed a significant drop in both accuracy and AUC. However, our proposed methodology showed that coefficient correlation does drops slowly, but maintains good performance in accuracy and AUC.

\figurename~\ref{fig:results_delta_kidney} shows an experiment for our proposed algorithm and DP algorithm with respect to kidney disease tests. Similar behavior was observed for the breast cancer dataset, but the coefficient correlation did not degrade after $\lambda=0.3$. In the case of AnomiGAN, we noticed from \figurename~\ref{fig:results_delta_breast} that coefficient correlation only drops until certain level due to additional loss from discriminator as described in Equation~\ref{loss1_eq}.

%Besides, we examined feature correlations for each dataset.
\figurename~\ref{fig:heatmap_both} details features and its correlation.
\figurename~\ref{fig:heatmap_both} (a) shows the feature correlation of breast cancer, and \figurename~\ref{fig:heatmap_both} (2) shows the feature correlation of chronic kidney disease. In the case of DP, we noticed that trade-off between privacy and accuracy are more strongly correlated if the features are strongly correlated with each other.

\begin{figure*}
\centering
\includegraphics[width=1\textwidth]{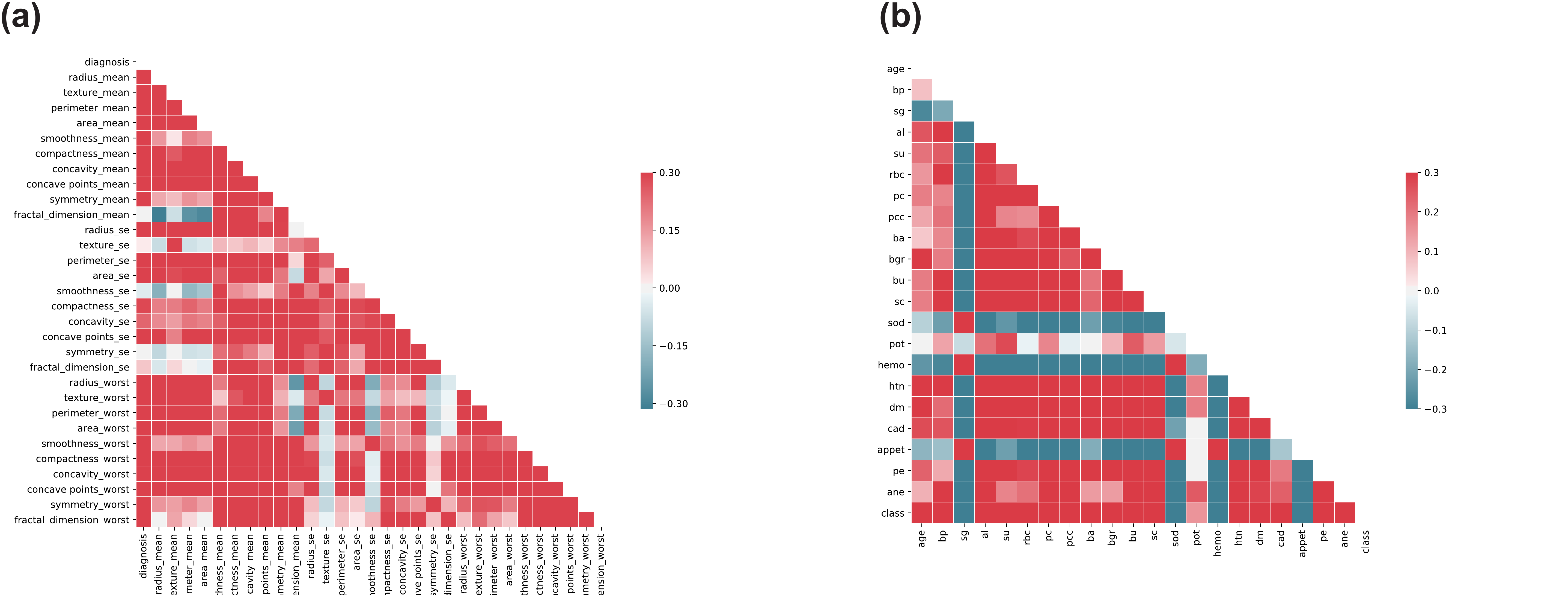}
\caption{Feature correlation (best viewed in color).
Correlations between features of breast cancer (a) and chronic kidney disease datasets (b). 
}
\label{fig:heatmap_both}
\end{figure*}

\iffalse
\begin{figure*}
\centering
\includegraphics[width=0.775\textwidth]{figures/loss_both2.pdf}
\caption{Learning performance of AnomiGAN.
(a) Loss of encoder and discriminator for breast cancer dataset.
(b) Loss of encoder and discriminator for chronic kidney disease dataset.
}
\label{fig:boss_loss}
\end{figure*}
\fi

\begin{figure*}
\centering
\includegraphics[width=1\textwidth]{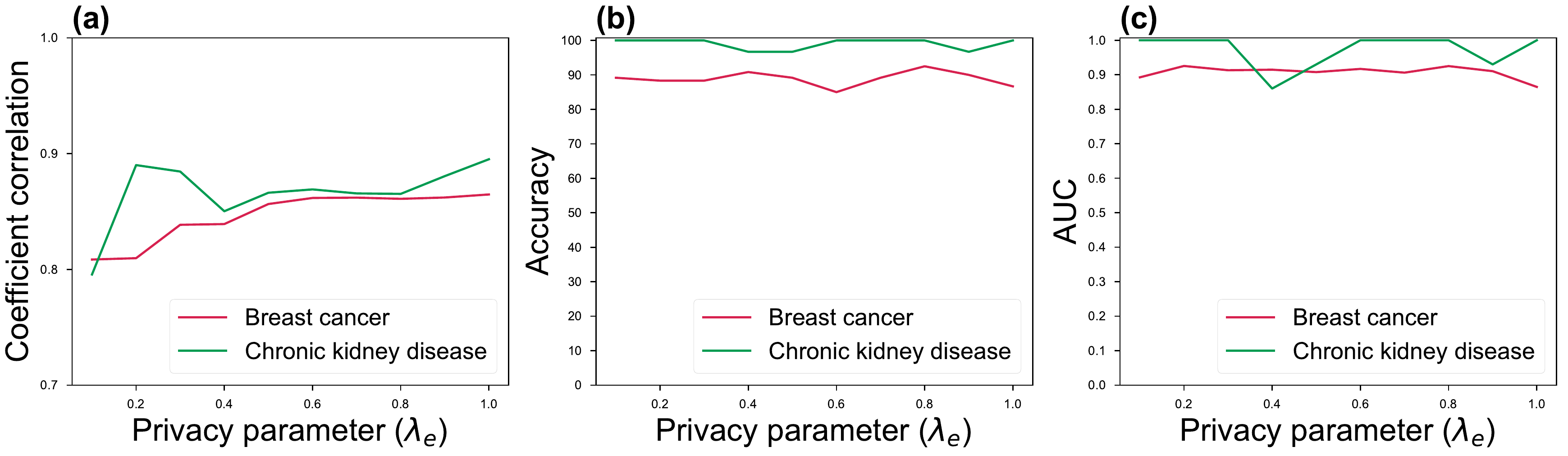}
\caption{Comparison of privacy parameter $\lambda_e$.
The correlation coefficient, accuracy, and AUC are measured by changing 0.1 of the privacy parameter  $\lambda_e$ for fixed test data.}
\label{fig:results_by_lambda}
\end{figure*}

\begin{spacing}{\mytablespacing}
\setlength{\tabcolsep}{2pt}
\ctable[
    caption = {Performance results of the model upon adding variance to the layer.},
    label = {table:cc_layers},
    doinside = \scriptsize,
    width = \columnwidth
]
{lccccccc}
{}
{
   \toprule
       & Layer1 & Layer2 & Layer3 & Layer4 & Layer5 & Layer6 & Layer7 \\      
     
    \midrule
     \textbf{(Breast cancer)} \\
     Correlation coefficient & 0.88& 0.86 & 0.87 & 0.83 & 0.83 & 0.86 & 0.86 \\
     Accuracy (\%) & 88.00 & 90.83 & 88.33 & 91.67 & 85.83 & 91.67 & 90.00 \\
     ACU & 0.864 & 0.916 & 0.858 & 0.906 & 0.931 & 0.921 & 0.906 \\
     
     \hline
     \textbf{(Chronic kidney disease)} \\
     Correlation coefficient & 0.86 &0.89 & 0.88 & 0.89 & 0.88 & 0.89 & 0.86 \\
     Accuracy (\%) & 96.70 & 100 & 100 & 100 & 100 & 100 & 100 \\
     ACU & 1.000 & 1.000 & 1.000 & 1.000 & 1.000 & 1.000 & 1.000 \\
    \bottomrule
}
\end{spacing}

%mention ration of lambda parameter!
\subsection{Performance Comparison}
We evaluated the performance of our proposed method based on two classifiers (breast cancer, chronic kidney disease) to measure prediction performance and coefficient values between original data and anonymized data. The experiments were conducted by changing 0.1 of the privacy parameters $\lambda_e$, and $\delta$.~\figurename~\ref{fig:results_by_lambda} shows an experiment for privacy parameter $\lambda_e$. 
The term $\lambda_e$ is directly associated with the Euclidean distance indicating that the privacy level should be decreased as $\lambda_e$ increases. As shown in \figurename~\ref{fig:results_by_lambda} (a) the coefficient correlation value indicates that strong association between original and anonymized data as the parameter $\lambda_e$ is increased for both datasets. \figurename~\ref{fig:results_by_lambda} (b) and (c) indicate that the averaged accuracy and AUC does not degrade as the parameter $\lambda_e$ is increased. This is expected behavior, as $\lambda_d$ is proportional to $\lambda_e$, maximizing the discriminator loss to a target classifier.

To validate the probabilistic behavior of our model with respect to variance added for each layer, we measured the mean of the coefficient correlation for each of the 7 encoder layers as shown in Table~\ref{table:cc_layers}. The results indicate that adding variance to each of layers has a  difference in the coefficient correlation with limited effects on both accuracy and AUC. 

Table~\ref{table:computation_time} shows the training and running time for each datasets. The training time varied depending on the hyperparameters. The training time was measured with the optimal hyperparameters. The loss weight parameters of $\lambda_e$ and $\lambda_d$ are set to 0.5, and the privacy parameter $\delta$ was set to 0.3. The training time includes the operation of measuring the variance of each layer which is then used at running time.

%Table~\ref{table:computation_time} shows the training (50000 epochs) and running time.

\begin{spacing}{\mytablespacing}
\setlength{\tabcolsep}{2pt}
\ctable[
    caption = {Training and running time of the proposed method based for two classifiers.},
    label = {table:computation_time},
    doinside = \scriptsize,
    width = \columnwidth
]
{lccc}
{}
{
   \toprule
      Dataset & No. of features & Training Time (Hours) & Running Time (Secs)\\      
     
    \midrule
     Breast cancer & 30 & 2 $\leq$ & 1.00 $\leq$  \\
     Chronic kidney disease & 24 &  1 $\leq$ & 1.00 $\leq$  \\

    \bottomrule
}
\end{spacing}

\end{methods}

\begin{methods}
    \section{Discussion}

Here, we have introduced a novel approach for anonymizing private data while preserving the original prediction accuracy. We showed that under a certain level of privacy parameters, our approach preserves privacy while maintaining a better performance of accuracy and AUC compared to the differential privacy. Moreover, we provide a mathematical overview showing that our model is secure against an efficient adversary, demonstrating that the estimated behavior of the model, and finally evaluated the performance compared to the state-of-the-art privacy preserving method.

One of our primary motivations for this study was that many companies are providing new services based on deep neural networks, and we believe this will extend to online medical services. Potential risks regarding the security of medical information (including genomic data) are higher compared to the current risks to private information security, as demonstrated by Facebook's recent privacy scandal. In addition, it is difficult to notice a privacy breach even when there are privacy policies in place. For example, when a patient consents to the use of medical diagnostic techniques, the propagation of that information to a third party cannot guarantee that the same privacy policies will be adhered to by them. Finally, machine learning as a service (MLaaS) is mostly provided by Google, Microsoft, or Amazon due to hardware constraints, and it is even more challenging to maintain user data privacy when using such services.

Exploiting traditional security in the deep learning requires encryption and decryption phases, which make its use impractical in the real world due to a vast amount of computation complexity. As a result, other privacy preserving techniques such as DP will be exploited in deep learning. Towards this objective, we developed a new approach of privacy-preserving technique based on deep learning. Our method is not limited to the medical data. Our framework can be extended in many various ways to the concept of exploiting a target classifier as a discriminator.

Unlike a statistics-based approach, our method does not require a background population to achieve good prediction results. AnomiGAN also provides the ability to share data while minimizing privacy risks. We believe that online medical services using the deep neural network technology will be available in our daily lives, and it will no longer be possible to overlook issues regarding the privacy of medical data. We believe that our methodology will encourage the anonymization of personal medical data. As part of future studies, we plan to extend our model to genomic data. The continuous investigation of privacy in medical data will benefit human health and enable the development of various diagnostic tools for early disease detection.
\end{methods}

%\medskip
%\small
\bibliographystyle{natbib}
\bibliography{references}

\begin{thebibliography}{}

\bibitem[Bae {\em et~al.}(2018)Bae, Jang, Jung, Jang, Ha, and
  Yoon]{bae2018security}
Bae, H., et~al. (2018).
\newblock Security and Privacy Issues in Deep Learning.
\newblock {\em arXiv preprint arXiv:1807.11655\/}.

\bibitem[Baluja(2017)Baluja]{baluja2017hiding}
Baluja, S. (2017).
\newblock Hiding images in plain sight: Deep steganography.
\newblock In {\em Advances in Neural Information Processing Systems\/}, pages
  2069--2079.

\bibitem[Berhane~Russom(2012)Berhane~Russom]{berhane2012concepts}
Berhane~Russom, M. (2012).
\newblock {\em Concepts of Privacy at the Intersection of Technology and
  Law\/}.
\newblock Ph.D. thesis.

\bibitem[Blake(1998)Blake]{blake1998uci}
Blake, C. (1998).
\newblock UCI repository of machine learning databases.
\newblock {\em http://www. ics. uci. edu/\~{} mlearn/MLRepository. html\/}.

\bibitem[Canetti {\em et~al.}(2004)Canetti, Goldreich, and
  Halevi]{canetti2004random}
Canetti, R., et~al. (2004).
\newblock The random oracle methodology, revisited.
\newblock {\em Journal of the ACM (JACM)\/}, {\bf 51}(4), 557--594.

\bibitem[Collins and Mansoura(2001)Collins and Mansoura]{collins2001human}
Collins, F.~S. et~al. (2001).
\newblock The human genome project: revealing the shared inheritance of all
  humankind.
\newblock {\em Cancer: Interdisciplinary International Journal of the American
  Cancer Society\/}, {\bf 91}(S1), 221--225.

\bibitem[Crotty and Slack(2016)Crotty and Slack]{crotty2016designing}
Crotty, B.~H. et~al. (2016).
\newblock Designing online health services for patients.
\newblock {\em Israel journal of health policy research\/}, {\bf 5}(1), 22.

\bibitem[Dwork(2008)Dwork]{dp_survey}
Dwork, C. (2008).
\newblock Differential privacy: A survey of results.
\newblock In {\em International Conference on Theory and Applications of Models
  of Computation\/}, pages 1--19. Springer.

\bibitem[Dwork(2011)Dwork]{dwork2011differential}
Dwork, C. (2011).
\newblock Differential privacy.
\newblock In {\em Encyclopedia of Cryptography and Security\/}, pages 338--340.
  Springer.

\bibitem[Dwork and Pottenger(2013)Dwork and Pottenger]{dwork2013toward}
Dwork, C. et~al. (2013).
\newblock Toward practicing privacy.
\newblock {\em Journal of the American Medical Informatics Association\/}, {\bf
  20}(1), 102--108.

\bibitem[Dwork {\em et~al.}(2014)Dwork, Roth, {\em
  et~al.}]{dwork2014algorithmic}
Dwork, C., et~al. (2014).
\newblock The algorithmic foundations of differential privacy.
\newblock {\em Foundations and Trends{\textregistered} in Theoretical Computer
  Science\/}, {\bf 9}(3--4), 211--407.

\bibitem[Erlich and Narayanan(2014)Erlich and Narayanan]{erlich2014routes}
Erlich, Y. et~al. (2014).
\newblock Routes for breaching and protecting genetic privacy.
\newblock {\em Nature Reviews Genetics\/}, {\bf 15}(6), 409.

\bibitem[Gilad-Bachrach {\em et~al.}(2016)Gilad-Bachrach, Dowlin, Laine,
  Lauter, Naehrig, and Wernsing]{gilad2016cryptonets}
Gilad-Bachrach, R., et~al. (2016).
\newblock Cryptonets: Applying neural networks to encrypted data with high
  throughput and accuracy.
\newblock In {\em International Conference on Machine Learning\/}, pages
  201--210.

\bibitem[Goodfellow {\em et~al.}(2014)Goodfellow, Pouget-Abadie, Mirza, Xu,
  Warde-Farley, Ozair, Courville, and Bengio]{goodfellow2014generative}
Goodfellow, I., et~al. (2014).
\newblock Generative adversarial nets.
\newblock In {\em Advances in neural information processing systems\/}, pages
  2672--2680.

\bibitem[Hesamifard {\em et~al.}(2017)Hesamifard, Takabi, and
  Ghasemi]{hesamifard2017cryptodl}
Hesamifard, E., et~al. (2017).
\newblock CryptoDL: Deep Neural Networks over Encrypted Data.
\newblock {\em arXiv preprint arXiv:1711.05189\/}.

\bibitem[Homer {\em et~al.}(2008)Homer, Szelinger, Redman, Duggan, Tembe,
  Muehling, Pearson, Stephan, Nelson, and Craig]{homer2008resolving}
Homer, N., et~al. (2008).
\newblock Resolving individuals contributing trace amounts of DNA to highly
  complex mixtures using high-density SNP genotyping microarrays.
\newblock {\em PLoS genetics\/}, {\bf 4}(8), e1000167.

\bibitem[Ioffe and Szegedy(2015)Ioffe and Szegedy]{ioffe2015batch}
Ioffe, S. et~al. (2015).
\newblock Batch normalization: Accelerating deep network training by reducing
  internal covariate shift.
\newblock {\em arXiv preprint arXiv:1502.03167\/}.

\bibitem[Kim {\em et~al.}(2016)Kim, Ha, Chun, Yoon, and
  Cha]{kim2016collaborative}
Kim, J., et~al. (2016).
\newblock Collaborative analytics for data silos.
\newblock In {\em Data Engineering (ICDE), 2016 IEEE 32nd International
  Conference on\/}, pages 743--754. IEEE.

\bibitem[Kingma and Ba(2014)Kingma and Ba]{kingma2014adam}
Kingma, D. et~al. (2014).
\newblock Adam: A method for stochastic optimization.
\newblock {\em arXiv preprint arXiv:1412.6980\/}.

\bibitem[LeCun {\em et~al.}(2012)LeCun, Bottou, Orr, and
  M{\"u}ller]{lecun2012efficient}
LeCun, Y.~A., et~al. (2012).
\newblock Efficient backprop.
\newblock In {\em Neural networks: Tricks of the trade\/}, pages 9--48.
  Springer.

\bibitem[Lindell and Katz(2014)Lindell and Katz]{lindell2014introduction}
Lindell, Y. et~al. (2014).
\newblock {\em Introduction to modern cryptography\/}.
\newblock Chapman and Hall/CRC.

\bibitem[Min {\em et~al.}(2017)Min, Lee, and Yoon]{min2017deep}
Min, S., et~al. (2017).
\newblock Deep learning in bioinformatics.
\newblock {\em Briefings in bioinformatics\/}, {\bf 18}(5), 851--869.

\bibitem[Nair and Hinton(2010)Nair and Hinton]{nair2010rectified}
Nair, V. et~al. (2010).
\newblock Rectified linear units improve restricted boltzmann machines.
\newblock In {\em Proceedings of the 27th international conference on machine
  learning (ICML-10)\/}, pages 807--814.

\bibitem[Oprisanu and De~Cristofaro(2018)Oprisanu and
  De~Cristofaro]{oprisanu2018anonimme}
Oprisanu, B. et~al. (2018).
\newblock AnoniMME: Bringing Anonymity to the Matchmaker Exchange Platform for
  Rare Disease Gene Discovery.
\newblock {\em bioRxiv\/}, page 262295.

\bibitem[Rubini and Eswaran(2015)Rubini and Eswaran]{rubini2015generating}
Rubini, L.~J. et~al. (2015).
\newblock Generating comparative analysis of early stage prediction of Chronic
  Kidney Disease.
\newblock {\em International Journal of Modern Engineering Research (IJMER)\/},
  {\bf 5}(7), 49--55.

\bibitem[Sankararaman {\em et~al.}(2009)Sankararaman, Obozinski, Jordan, and
  Halperin]{sankararaman2009genomic}
Sankararaman, S., et~al. (2009).
\newblock Genomic privacy and limits of individual detection in a pool.
\newblock {\em Nature genetics\/}, {\bf 41}(9), 965.

\bibitem[Sanyal {\em et~al.}(2018)Sanyal, Kusner, Gasc{\'o}n, and
  Kanade]{sanyal2018tapas}
Sanyal, A., et~al. (2018).
\newblock TAPAS: Tricks to Accelerate (encrypted) Prediction As a Service.
\newblock {\em arXiv preprint arXiv:1806.03461\/}.

\bibitem[Schuster(2007)Schuster]{schuster2007next}
Schuster, S.~C. (2007).
\newblock Next-generation sequencing transforms today's biology.
\newblock {\em Nature methods\/}, {\bf 5}(1), 16.

\bibitem[Simmons and Sahinalp(2019)Simmons and Sahinalp]{sean2019}
Simmons, Sean, B.~B. et~al. (2019).
\newblock Protecting Genomic Data Privacy with Probabilistic Modeling.
\newblock {\em Proceedings of the 24th Pacific Symposium on Biocomputing\/}.

\bibitem[Simmons and Berger(2015)Simmons and Berger]{simmons2015one}
Simmons, S. et~al. (2015).
\newblock One size doesn’t fit all: measuring individual privacy in aggregate
  genomic data.
\newblock In {\em Proceedings. IEEE Symposium on Security and Privacy.
  Workshops\/}, volume 2015, page~41. NIH Public Access.

\bibitem[Simmons and Berger(2016)Simmons and Berger]{simmons2016realizing}
Simmons, S. et~al. (2016).
\newblock Realizing privacy preserving genome-wide association studies.
\newblock {\em Bioinformatics\/}, {\bf 32}(9), 1293--1300.

\bibitem[Simmons {\em et~al.}(2016)Simmons, Sahinalp, and
  Berger]{simmons2016enabling}
Simmons, S., et~al. (2016).
\newblock Enabling privacy-preserving GWASs in heterogeneous human populations.
\newblock {\em Cell systems\/}, {\bf 3}(1), 54--61.

\bibitem[Wagner and Eckhoff(2018)Wagner and Eckhoff]{wagner2018technical}
Wagner, I. et~al. (2018).
\newblock Technical privacy metrics: a systematic survey.
\newblock {\em ACM Computing Surveys (CSUR)\/}, {\bf 51}(3), 57.

\bibitem[Weir {\em et~al.}(2004)Weir, Olick, Murray, {\em
  et~al.}]{weir2004stored}
Weir, R.~F., et~al. (2004).
\newblock {\em The stored tissue issue: Biomedical research, ethics, and law in
  the era of genomic medicine\/}.
\newblock Oxford University Press.

\bibitem[Zhou {\em et~al.}(2011)Zhou, Peng, Li, Chen, Tang, and
  Wang]{zhou2011release}
Zhou, X., et~al. (2011).
\newblock To release or not to release: evaluating information leaks in
  aggregate human-genome data.
\newblock In {\em European Symposium on Research in Computer Security\/}, pages
  607--627. Springer.

\end{thebibliography}
		
\end{document}